\documentclass[11pt]{amsart}
\usepackage[a4paper, total={6in, 8in}]{geometry}

\usepackage{enumerate}
\usepackage{amsmath}
\usepackage{amssymb,latexsym}
\usepackage{amsthm}
\usepackage{color}
\usepackage{cancel}
\usepackage{graphicx}
\usepackage{cite}
\usepackage{changes}
\usepackage{hyperref}
\usepackage{comment}
\usepackage{amscd}
\usepackage{mathtools}

\newtheorem{theorem}{Theorem}[section]

\newtheorem{lemma}[theorem]{Lemma}
\newtheorem{corollary}[theorem]{Corollary}
\newtheorem{definition}[theorem]{Definition}
\newtheorem{proposition}[theorem]{Proposition}

\newtheorem{remark}[theorem]{Remark}

\DeclareMathOperator{\C}{\mathcal{C}}

\newcommand{\fqn}{\mathbb{F}_{q^n}}
\newcommand{\fq}{\mathbb{F}_{q}}
\newcommand{\Z}{\mathbb{Z}}

\newcommand{\F}{{\mathbb F}}

\newcommand{\w}{{\mathbf w}}

\renewcommand{\mod}{\hbox{{\rm mod}\,}}

\newcommand{\PG}{\mathrm{PG}}
\newcommand{\N}{\mathrm{N}}

\newcommand{\xx}{\mathbf{x}}

\newcommand{\SSS}{\mathcal{S}}

\newcommand\qbin[3]{\left[\begin{matrix} #1 \\ #2 \end{matrix} \right]_{#3}}

\begin{document}

\title[Using multi-orbit cyclic subspace codes for OOCs]{Using multi-orbit cyclic subspace codes for constructing optical orthogonal codes}

\author[F. {\"O}zbudak]{Ferruh {\"O}zbudak}
\address{Ferruh {\"O}zbudak, \textnormal{FENS, Sabancı University, Ístanbul, Turkey}}
\email{ferruh.ozbudak@sabanciuniv.edu}

\author[P. Santonastaso]{Paolo Santonastaso}
\address{Paolo Santonastaso, \textnormal{Dipartimento di Matematica e Fisica, Universit\`a degli Studi della Campania ``Luigi Vanvitelli'', Caserta, Italy}}
\email{paolo.santonastaso@unicampania.it}

\author[F. Zullo]{Ferdinando Zullo}
\address{Ferdinando Zullo, \textnormal{Dipartimento di Matematica e Fisica, Universit\`a degli Studi della Campania ``Luigi Vanvitelli'', Caserta, Italy}}
\email{ferdinando.zullo@unicampania.it}

\begin{abstract}
We present a new application of multi-orbit cyclic subspace codes to construct large optical orthogonal codes, with the aid of the multiplicative structure of finite fields extensions. This approach is different from earlier approaches using combinatorial and additive (character sum) structures of finite fields. Consequently, we immediately obtain new classes of optical orthogonal codes with different parameters.
\end{abstract}

\maketitle

\noindent {\textbf{Keywords}: optical orthogonal code; cyclic subspace code; linearized polynomial}\\

\noindent{\textbf{MSC2020}: 94B60; 11T30; 11T71}

\section{Introduction}

In their groundbreaking paper \cite{koetter2008coding}, Koetter and Kschischang established the foundations of subspace codes, which turns out to be useful for the error correction in random network coding. Let $\F_q$ be a finite field and $m \geq 2$ be an integer. We consider all $\F_q$-subspaces of $\F_{q^m}$. If $U,V$ are two $\F_q$-subspaces of $\F_{q^m}$, then the \textbf{subspace distance} $d_s(U,V)$ is defined as 
\[
d_s(U,V)=\dim_{\F_q}(U)+\dim_{\F_q}(V)-2\dim_{\F_q}(U\cap V).
\]
In this paper we consider constant dimension $\F_q$-linear subspaces. Hence let $1 \leq k \leq m$ be a fixed integer.  
If $\dim_{\F_q}(U)=\dim_{\F_q}(V)=k$, then the subspace distance reads as $\dim_s(U,V)=2k-2\dim_{\F_q}(U \cap V)$. From now on, let $\mathcal{G}_{q}(m,k)$ denote the Grassmannian consisting of all $k$-dimensional $\F_q$-subspaces of $\F_{q^m}$. There is a natural action of the multiplicative group $\F_{q^m}^*$ on $\mathcal{G}_{q}(m,k)$: for $c \in \F_{q^m}^*$ and $U \in \mathcal{G}_{q}(m,k)$, let $cU$ denote the action of $c$ on $U$, which is defined as 
\[
cU=\{cu\colon u \in U\} \in \mathcal{G}_{q}(m,k).
\]
Therefore a \textbf{cyclic subspace code} $C$ in $\mathcal{G}_{q}(m,k)$ is a collection of elements of $\mathcal{G}_{q}(m,k)$ such that $U \in C$ implies $cU \in C$, for every $c \in \F_{q^m}^*$, see \cite{etzion2011error}. It has been an active research area to construct cyclic subspace codes with best parameters in the recent years. One of the most useful techniques in the recent years for the construction of cyclic subspace codes is via Sidon spaces \cite{roth2017construction}. Here we present an equivalent definition: $U \in \mathcal{G}_{q}(m,k)$ is a \textbf{Sidon space} of dimension $k$ in $\F_{q^m}$ if and only if 
\[
\dim_{\F_q}(U \cap \alpha U) \leq 1, \mbox{ for each }\alpha \in \F_{q^m} \setminus \F_q.
\]
This suggests a natural extension. For $r \geq 1$, let $U_1,\ldots,U_r$ be distinct elements of $\mathcal{G}_{q}(m,k)$. We say that $\{U_1,\ldots,U_r\}$ is an \textbf{$r$-multi Sidon space} if 
\[
\dim_{\F_q}(U_i \cap \alpha U_j) \leq 1,
\]
for every $\alpha \in \F_{q^m}^*$ and $1 \leq i,j \leq r$ with $i \neq j$, and for every $\alpha \in \F_{q^m} \setminus \F_q$ and $i=j$, see \cite{zullo2023multi}. In particular, multi Sidon spaces have important applications in cyclic subspace codes.  In this paper, using a multiplicative approach we obtain a new connection of multi-orbit cyclic subspace and codes and multi Sidon spaces to optical orthogonal codes. 

Given a vector $(x_0,\ldots,x_{n-1}) \in \F_{2}^n$, and for an integer $i\geq n$ we denote by $x_i:=x_{i \pmod{n}}$.

\begin{definition}
    An $(n,w,\lambda)$ \textbf{optical orthogonal code} (\textbf{OOC} for short) $\C$ is an Hamming-metric code in $\F_2^n$ having constant weight $w$, satisfying the following two properties:
    \begin{enumerate}
        \item \textbf{The autocorrelation property:} 
        \[
        \sum_{t=0}^{n-1} x_tx_{t + \tau} \leq \lambda,
        \]
        for every $(x_0,\ldots,x_{n-1}) \in \C$ and every integer $\tau$, with $0<\tau<n$.
        \item \textbf{The crosscorrelation property:}
        \[
        \sum_{t=0}^{n-1} x_ty_{t + \tau} \leq \lambda,
        \]
        for every distinct $(x_0,\ldots,x_{n-1}),(y_0,\ldots,y_{n-1}) \in \C$ and every non negative integer $\tau$.
    \end{enumerate}
\end{definition}

Since every codeword $\xx \in \C$ has Hamming weight $w$, the autocorrelation equals $w$ in the extremal case $\tau=0$. To avoid triviality, we require that $\lambda<w$. 
For the applications of OOCs in optical code-division multiple access (OCDMA) communication systems, it is desirable that an OOC have large size and small maximum correlation for a given length $n$ and a given weight $w$; see \cite{salehi1989code,salehi1989code2}.
As also pointed out in \cite{chung1989optical}, OOCs when used in OCDMA communication systems need to have the property that each codeword has many more 0’s than 1’s. In terms of parameters, for an $(n,w,\lambda)$ OOC, this means that the parameter $n$ is significantly larger than $w$.
However, there is a tradeoff among the parameters of an OOC, known as \textbf{Johnson bound}.

\begin{theorem} [see \textnormal{\cite{johnson1962new}}] \label{thm:Jonbound}
    Let $\C$ be an $(n,w,\lambda)$ OOC. Then
    \begin{equation} \label{eq:johnsonbound}
    \lvert \C \rvert \leq J(n,w,\lambda):= \left\lfloor \frac{1}{w} \left\lfloor \frac{n-1}{w-1} \left\lfloor \frac{n-2}{w-2} \left\lfloor \cdots \left\lfloor \frac{n-\lambda}{w-\lambda} \right\rfloor \cdots \right\rfloor \right\rfloor \right\rfloor \right\rfloor.
     \end{equation}
\end{theorem}

When $\lambda=1$, the optimal set size of an $(n,w,1)$ OOC is given by
\[
\lvert \C \rvert = \left\lfloor \frac{1}{w} \left\lfloor \frac{n-1}{w-1}  \right\rfloor \right\rfloor.
\]
There exist several constructions of optimal and large non optimal OOCs with $\lambda=1$ in the literature, obtained by relying on other areas, such as projective geometry \cite{chung1989optical,miyamoto2004optical}, finite field theory \cite{chung1990optical,weng2001optical,ding2003several,ding2004cyclotomic}, design theory \cite{abel2004some,buratti2002cyclic,buratti2010further,buratti1995powerful,yang1995optical,yang1995some} and Sidon sets \cite{ruiz2020new}.  Other techniques also involve computer algorithms such as greedy and accelerated greedy algorithms \cite{chung1989optical} and outer-product matrix algorithm \cite{charmchi2006outer}. 

When $\lambda=2$, the optimal set size of an $(n,w,2)$ OOC is given by
\[
\lvert \C \rvert = \left\lfloor \frac{1}{w} \left\lfloor \frac{n-1}{w-1} \left\lfloor \frac{n-2}{w-2} \right\rfloor \right\rfloor \right\rfloor.
\]
The first construction of optimal OOC with $\lambda=2$ are the $(p^{2h}-1,p^{h+1},2)$ OOCs having size $p^h-2$, where $p$ is any prime and $h$ is a positive integer, derived from projective geometry and constructed in \cite{chung1990optical}. Moreover, in \cite{alderson2008families}, new constructions of optimal $(n,w,2)$ OOC, with $w \in \{4,6\}$ are obtained. These constructions rely on various techniques in finite projective spaces involving hyperovals in projective planes and Singer groups. Other constructions of large OOCs for $\lambda=2$ can be found in \cite{chu2004optimal,feng2008constructions}. 
When $\lambda > 2$, it is believed that the Johnson bound \eqref{eq:johnsonbound} is not achievable. Apart from the constructions found in \cite{alderson2008geometric}, where the authors have found optimal $(\frac{q^{h+1}-1}{q-1},\lambda+1,\lambda)$ OOCs satisfying $\gcd(h+1,i)=1$ for any $1 \leq i \leq \lambda$, there are no other known constructions.
Since for applications, it is very interesting to find OOCs with a large size with respect to their parameters, a well-studied problem in the area of OOCs is dedicated to construct codes that are close to being optimal, starting from the notion of asymptotic optimality introduced in \cite{moreno1995new}. If an $(n,w,\lambda)$ OOC $\C$ having size $A(n,w,\lambda)$ satisfies 
\[
\lim_{n \rightarrow + \infty} \frac{A(n,w,\lambda)}{J(n,w,\lambda)}=1,
\]
$\C$ is referred to as an \textbf{asymptotically optimal OOC}. When $\lambda=1$, there are constructions of asymptotically optimal OOCs, see e.g \cite{moreno1995new,ruiz2020new}. For the case $\lambda=2$, in \cite{miyamoto2004optical}, the authors, by using conics on finite projective planes in the projective geometry $\PG(3,q)$, where $q$ is a prime power, constructed a family of asymptotically optimal $(q^3+q^2+q+1,q+1,2)$ OOCs having size $q^3-q^2+q$. Other asymptotically optimal constructions when $\lambda=2$ are derived in \cite{moreno1995new}. For the parameters of known asymptotically optimal OOCs with $\lambda = 2$, please consult Table 2 of \cite{alderson2008geometric}. 
By using normal rational curves in $\PG(d, q)$, in \cite{alderson2007optical}, the authors constructed a family of asymptotically optimal $(\frac{q^{\lambda+2}-1}{q-1},q+1,\lambda)$ OOCs, for any choice of $\lambda>1$ and $q>\lambda$. In \cite{alderson2007optical2}, by using  normal rational curves and Singer groups, the authors constructed asymptotically optimal OOCs generalizing and improving previous constructions of OOCs, for instance those from conics \cite{miyamoto2004optical}
and arcs \cite{alderson2007optical}. Among the various large constructions of OOCs, the authors have constructed a family of asymptotically optimal designs denoted as $(\frac{q^{h+1}-1}{q-1},q+1,\lambda)$, where $m$ and $\lambda$ are integers satisfying $h > \lambda \geq 2$, and $q$ is a prime power such that $q \geq \lambda$ (refer to \cite[Theorem 2]{alderson2007optical2}). Additional constructions are presented in \cite{alderson2008geometric,moreno1995new,alderson2007optical2,alderson2007constructions,alderson2008classes,omrani2005improved}. 

We summarize in Table \ref{table:asynoptimalOOC} the list of the known largest constructions of OOCs having $\lambda>2$. Different rows of the tables correspond to different parameters $(n,w,\lambda)$ of the OOCs.

\begin{table}[htp]
\tabcolsep=0.2 mm
\begin{tabular}{|c|c| c |c|c|}
\hline
& \mbox{Parameters} & \mbox{Conditions} & \mbox{References} \\ \hline
1) & $\left(q^2-1,q-2,3\right)$ & none & \cite[Theorem 14]{alderson2008geometric}  \\ \hline
2) & $\left(\frac{q^5-1}{q-1},q,3\right)$ & none & \cite[Corollary 4.10]{alderson2007constructions} \\ \hline
3) & $\left(\frac{q^{k(\lambda+2)}-1}{q^k-1},q-1,\lambda\right)$ & $\lambda>1$ & \cite[Theorem 20]{alderson2008classes} \\ \hline
4) & $\left(\frac{q^{(h+1)}-1}{q-1},q+1,\lambda\right)$ & $h>\lambda>1$ and $q \geq \lambda$ & \cite[Theorem 3.4]{alderson2007optical2} \\ \hline
5) & $\left(q^h-1,q-\lambda+3,\lambda\right)$ & $h>\lambda \geq 3$ and $q \geq \lambda$ & \cite[Corollary 5.6]{alderson2007optical2} \\ \hline
6) & $\left(ph,h,\lambda\right)$ & $p$ prime and $h \mid (p-1)$ & \cite[Theorem 1]{moreno1995new} \\ \hline
7) & $\left((q-1)p,(p-\lambda),\lambda \right)$ & $q=p^a$ and $1\leq \lambda \leq p-\lambda$ & \cite[Theorem 2]{moreno1995new} \\ \hline
8) & $\left(h(q+1),h,2\lambda\right)$ & $q=p^a$, $p$ prime, $h\mid (q-1)$ and $\gcd(h,q+1)=1$ & \cite[Theorem 3]{moreno1995new} \\ \hline
9) & $\left(\frac{q^{\lambda-1}-1}{q-1},q+1,\lambda\right)$ & $\lambda \geq 4$ & \cite[Theorem 11 and 12]{alderson2008geometric} \\ \hline
10) & $\left(q^a-1,q^h,q^{h-1}\right)$ & $a>h\geq 1$ & \cite[Theorem 11]{omrani2005improved} \\ \hline
11) & $(q-1,f,\lambda)$ & $f \mid (q-1)$ and $\lambda <f$ & \cite[Theorem 1]{ding2003several} $^*$ \\
\hline
12) & $(3^{13f}-1,81,3)$ & none & \cite[Proposition 3]{irimaugzi2022two} $^*$ \\
\hline 
\end{tabular}
\caption{Known largest family of OOCs with $\lambda>2$. ($q$ is a prime power) \\
(Non asympotically optimal constructions are marked with $^*$)}
\label{table:asynoptimalOOC}
\end{table}


In this paper we will show a construction of OOCs with the aid of multi-orbit cyclic subspace codes or multi-Sidon spaces.
Consequently, we construct new classes of optical orthogonal codes with new parameters. The rest of the paper is organized as follows. We explain first the connection to multiplicative structures and finite field extensions in Section \ref{sec:OOCfromextension}. Furthermore, we present our constructions of new classes of optical orthogonal codes using a connection to cyclic subspace codes and to multi-Sidon spaces in Section \ref{sec:constructionscyclic}. Finally, in Section \ref{sec:new} we show infinitely many new parameters for which we are able to construct OOCs. 

\section{OOC codes from extension fields}\label{sec:OOCfromextension}

In this section we will see how the problem of finding OOCs can be translated in terms of sets of integers and then we use this correspondence to describe OOCs in extension fields.
Let us start by observing that it is possible to establish a one-to-one correspondence between elements of $\F_{2}^n$ and subsets of $\Z_n$ in the following way:
\[
\begin{array}{lcrc}
     \theta: & \F_2^n & \longrightarrow & 2^{\Z_n}  \\
     & \mathbf{x}=(x_0,\ldots,x_{n-1}) & \longmapsto & \theta(\mathbf{x})=\{i \in \Z_n \colon x_i=1\},
\end{array}
\]
so $\theta$ is the map that associates each element of $\F_2^n$ with its support (with respect to the Hamming metric).
Note also that $\theta$ is invertible and $\w_H(\mathbf{x})=\lvert \theta(\mathbf{x}) \rvert$, where $\w_H(\mathbf{x})$ denotes the Hamming weight of $\mathbf{x}$.  For a non negative integer $\tau$, and a non-empty subset $X \subseteq \Z_n$ we define \[X+\tau:=\{a+\tau \colon a \in X \} \subseteq \Z_n.\] 

The correspondence given by $\theta$ allows us to translate the definition of an OOC in terms of subsets of $\Z_n$.

\begin{definition} \label{def:oos}
An $(n,w,\lambda)$ \textbf{Optical Orthogonal Set} (OOS) $\mathcal{S}$ is a subset of $2^{\Z_n}$ with the property that $\lvert X \rvert =w$ for every $X \in \mathcal{S}$, and such that the following two properties hold:
\begin{enumerate}
        \item The autocorrelation property: 
        \[
        \lvert X \cap (X + \tau) \rvert \leq \lambda,
        \]
        for every $X \in \SSS$ and every integer $\tau$, with $0<\tau<n$.
        \item The crosscorrelation property:
        \[
        \lvert X \cap (Y + \tau) \rvert  \leq \lambda,
        \]
        for every distinct $X,Y \in \SSS$ and every non negative integer $\tau$.
    \end{enumerate}
\end{definition}

As a consequence of the properties of $\theta$, we have the following correspondence, showing that constructing OOCs is equivalent to construct sets as defined in Definition \ref{def:oos} (OOS).

\begin{lemma} \label{lm:connOOCOOS}
    Let $\C$ be an $(n,w,\lambda)$ OOC. Then $\theta(\C)$ is an $(n,w,\lambda)$ OOS. Conversely, let $\SSS$ be an $(n,w,\lambda)$ OOS. Then $\theta^{-1}(\SSS)$ is an $(n,w,\lambda)$ OOC. 
    \end{lemma}

In the next sections, we will provide new families of OOCs by showing new constructions of OOSs.
To this aim, we show a procedure to construct OOSs, by using subsets of a finite field. Let $\omega$ be a primitive element of $\F_{q^m}$. Let $W$ be a nonempty subset of $\F_{q^m}$. Denote by 
\[
\SSS(W):=\{i \colon \omega^i \in W\} \subseteq \Z_N
\]
with $N=q^m-1$. Note that $\lvert S(W) \rvert=\lvert W^* \rvert$, where $W^*:=W \setminus \{0\}$ and that 
\begin{equation} \label{eq:shift=transl}
\SSS(\alpha W)=\SSS(W)+\tau,
\end{equation}
if $\alpha=\omega^{\tau}$. 

In the following result we will see how to construct OOSs, using the multiplicative properties of subsets of $\F_{q^m}$.

\begin{theorem}\label{thm:defOOS}
    Let $W_1,\ldots,W_t$ be distinct nonempty subsets of $\F_{q^m}^*$ and let $N=q^m-1$. Define
    \[
    \SSS(W_1,\ldots,W_t):=\{\SSS(W_i)\colon i\in\{1,\ldots,t\}\} \subseteq 2^{\Z_N}.
    \]
    The set $\SSS(W_1,\ldots,W_t)$ is an $(q^m-1,w,\lambda)$ OOS if and only if $\lvert W_i^* \rvert=w$ for every $i$, and 
    \begin{enumerate}
        \item for every $i\in\{1,\ldots,t\}$  and every $\alpha \in \F_{q^m} \setminus \{0,1\}$
        \[ \lvert W_i \cap \alpha W_i \rvert \leq \lambda;\]
        \item for every $i,j \in\{1,\ldots,t\}$ with $i\ne j$ and every $\alpha \in \F_{q^m}^*$
        \[ \lvert W_i \cap \alpha W_j \rvert \leq \lambda.\]
    \end{enumerate}
    In particular, $\lvert \SSS(W_1,\ldots,W_t) \rvert =t$.
\end{theorem}
\begin{proof}
    By Definition \ref{def:oos}, $\SSS(W_1,\ldots,W_t)$ is an $(N,w,\lambda)$ if and only if every set of $\SSS(W_1,\ldots,W_t)$ has size $w$ and 
    \begin{itemize}
        \item for any $i \in \{1,\ldots,t\}$ and for any $\tau\in \{1,\ldots,q^m-2\}$ then $|\SSS(W_i)\cap(\SSS(W_i)+\tau)|=|W_i\cap \omega^\tau W_i|\leq \lambda$;
        \item for any $i,j \in \{1,\ldots,t\}$ with $i\ne j$ and for any $\tau\in \{1,\ldots,q^m-2\}$ then $|\SSS(W_i)\cap(\SSS(W_j)+\tau)|=|W_i\cap \omega^\tau W_j|\leq \lambda$,
    \end{itemize}
    and so this happens if and only if the $W_i$'s satisfy (1) and (2).
\end{proof}

In the next section we will see that if we choose the $W_i$'s as the affine spaces associated with the distinct representative of cyclic subspace codes we obtain an OOS.

\section{Constructions from cyclic subspace codes}\label{sec:constructionscyclic}

Let $C$ be a subset of $\mathcal{G}_{q}(m,k)$, we recall that the \textbf{minimum subspace distance} of a subspace code $C$ is defined as 
\[
d_s(C) := \min\{d(U,V) \colon  U,V \in C, U \neq V\}.
\]
As already seen in the Introduction, a subspace code $C \subseteq \mathcal{G}_q(m,k)$ is said to be \textbf{cyclic} if for every $\alpha \in \F_{q^m}^*$ and every $V \in C$ then $\alpha V \in C$.  
A \textbf{one-orbit (cyclic) code} is a subspace code of the form \[C_U:=\{\alpha U: \alpha \in \F_{q^m}^*\} \subseteq \mathcal{G}_{q}(m,k),\]
where $U$ is an $\F_q$-subspace of $\F_{q^m}$ such that $\dim_{\F_q}(U)=k$. Sometimes, $C_U$ is also refereed as the \textbf{orbit} of $U$. The size of $C_U$ is $(q^m-1)/(q^t-1)$, for some $t$ which divides $m$. The largest size of such a code is $(q^m-1)/(q-1)$, and codes attaining this size are called \textbf{full-length one-orbit codes}. Full-length orbit codes with distance $2k - 2$, which is the best possible, will be called \textbf{optimal full-length one-orbit codes}. These kinds of subspace codes are connected with Sidon spaces.
An $\fq$-subspace $U$ of $\fqn$ is called a \textbf{Sidon space} if the product of any two elements of $U$ has a unique factorization over $U$, up to multiplying by some elements in $\fq$.
More precisely, $U$ is a Sidon space if for all nonzero $a,b,c,d \in U$, if $ab=cd$, then 
\[ \{a\fq, b\fq\}=\{c\fq, d\fq\}, \]
where if $e\in \F_{q^m}$ then $e\fq=\{e \lambda \colon \lambda \in \fq\}$.

\begin{theorem} [see \textnormal{\cite[Lemma 34]{roth2017construction}}] \label{lem:charSidon2}
Let $U$ be an $\fq$-subspace of $\F_{q^m}$ of dimension $k$. Then $C_U$ is an optimal full-length one-orbit code if and only if $U$ is a Sidon space.
\end{theorem}

The above connection can be extended to multi-orbit cyclic subspace codes.
Let $U_1,\ldots,U_r$ be $\fq$-subspaces of dimension $k$ in  $\F_{q^m}$ and let
\begin{equation}\label{eq:multiorbitcodes} C=\bigcup_{i \in \{1,\ldots,r\}} C_{U_i} \subseteq \mathcal{G}_{q}(m,k)
\end{equation}
with minimum distance $2k-2$ and of size $r\frac{q^m-1}{q-1}$.
These codes have the same minimum distance of the codes associated with Sidon spaces but they are larger.

The code $C$ is uniquely defined by the subspaces $U_1,\ldots,U_r$ and therefore we give the following two definitions.
The set $\{U_1,\ldots,U_r\}$ for the cyclic subspace code  \eqref{eq:multiorbitcodes} is said to be a set of \textbf{representatives} of $C$, whereas if also its minimum distance is $2k-2$ we will call it a \textbf{multi-Sidon space}, in connection with the correspondence between cyclic subspace codes with these parameters with one orbit and Sidon spaces. 

We will now show how to use cyclic subspace codes to construct OOCs. First we describe the intersection properties of affine spaces associated with the subspaces of a cyclic subspace code.

\begin{proposition}\label{prop:connOOScyclicsubspacecodes}
Let $C=\cup_{i \in \{1,\ldots,r\}}C_{U_i} \subseteq \mathcal{G}_q(m,k)$ be a cyclic subspace code with minimum distance $d$ and with the property that the $C_{U_i}$'s are $r$ distinct orbits. For every $i \in \{1,\ldots,r\}$, let $d_{i,1},\ldots,d_{i,t}$ be pairwise not $\F_q$-proportional representatives of the elements of the quotient $\F_{q^m}/U_i$, with $t=(q^{m-k}-1)/(q-1)$.
Let
\[\mathcal{A}(U_1,\ldots,U_r)=\bigcup_{i=1}^r \{U_i+d_{i,j} \colon j\in\{1,\ldots,t\}\}.\]
Then 
\[|\mathcal{A}(U_1,\ldots,U_r)|=t\cdot r,\]
and 
    \begin{enumerate}
        \item for every $V \in \mathcal{A}(U_1,\ldots,U_r)$  and every $\alpha \in \F_{q^m} \setminus \{0,1\}$
        \[ \lvert V \cap \alpha V \rvert \leq q^{k-d/2};\]
        \item for every $V_1,V_2 \in \mathcal{A}(U_1,\ldots,U_r)$ with $V_1\ne V_2$ and every $\alpha \in \F_{q^m}^*$
        \[ \lvert V_i \cap \alpha V_j \rvert \leq q^{k-d/2}.\]
    \end{enumerate}
\end{proposition}
\begin{proof}
    For any subspace $U$ in $C$ define the set of affine subspaces associated with $U$
    \[ A(U)=\{ U+a \colon a \in \F_{q^m}^* \}. \]
    It is easy to see that $|\mathcal{A}(U)|=|A(U)|/(q-1)$ and $|A(U)|=q^{m-k}-1$. Since for $U_i,U_j$, with $i \neq j$ we have that $A(U_i)\cap A(U_j)=\emptyset$, we have that
    \[
    \lvert \mathcal{A}(U_1,\ldots,U_r) \rvert = \sum_{i=1}^r \lvert \mathcal{A}(U_i) \rvert=r \frac{q^{m-k}-1}{q-1},
    \]
    that proves the first part of the assertion. For the second part, let $V\in \mathcal{A}(U_1,\ldots,U_r)$. So, $V=U_i+d_{i,j}$, for some $i \in \{1,\ldots,r\}$ and $j \in \{1,\ldots,t\}$. Let $\alpha \in \F_{q^m} \setminus \{0,1\}$. If $\alpha \notin \F_q$, then 
    \[
    \lvert V \cap \alpha V \rvert=\lvert (U_i+d_{i,j}) \cap (\alpha U_i + \alpha d_{i,j})\rvert \leq \lvert U_i \cap \alpha U_i \rvert \leq q^{k-d/2}.
    \]
    Otherwise, if $\alpha \in \F_q \setminus \{0,1\}$ then $ V \cap \alpha V= (U_i+d_{i,j}) \cap (U_i + \alpha d_{i,j})= \emptyset$. Indeed, if $ V \cap \alpha V= (U_i+d_{i,j}) \cap ( U_i + \alpha d_{i,j}) \neq \emptyset$, there exists $u,v \in U_i$ such that $u+d_{i,j}=v+\alpha d_{i,j}$, implying that $d_{i,j} \in U_i$, a contradiction to our hypothesis. This proves $(1)$. Now, let $V_1,V_2 \in \mathcal{A}(U_1,\ldots,U_r)$ with $V_1\ne V_2$. Assume, that $V_1=U_i+d_{i,j}$ and $V_2=U_i+d_{i,h}$, with $j \neq h$. If $\alpha \in \F_{q^m}^*$ then 
    \[ V_1 \cap \alpha V_2= (U_i+d_{i,j}) \cap (\alpha U_i + \alpha d_{i,h})= \emptyset.\] Indeed, if $(U_i+d_{i,j}) \cap (\alpha U_i + \alpha d_{i,h}) \neq  \emptyset$, there exists $u,v \in U_i$ such that $u+d_{i,j}=v+\alpha d_{i,h}$, implying that $d_{i,j}$ and $d_{i,h}$ are $\F_q$-proportional representative of the elements of the quotient $\F_{q^m}/U_i$, a contradiction to our hypothesis. While, if $V_1=U_i+d_{i,j_1}$ and $V_2=U_h+d_{i,j_2}$, with $i \neq h$, therefore 
    \[
    \lvert V_1 \cap \alpha V_2 \rvert=\lvert (U_i+d_{i,j_1}) \cap (\alpha U_h + \alpha d_{i,j_2})\rvert \leq \lvert U_i \cap \alpha U_h \rvert \leq q^{k-d/2}.
    \]
    This proves $(2)$ and the assertion.
\end{proof}

As a consequence we immediately obtain the following result, where we show how to construct OOCs from cyclic subspace codes. 

\begin{theorem} \label{th:OOCfromsubspace}
    Let $C=\cup_{i \in \{1,\ldots,r\}}C_{U_i} \subseteq \mathcal{G}_q(m,k)$ be a cyclic subspace code with minimum distance $d$ and with the property that the $C_{U_i}$'s are $r$ distinct orbits.
    Then $\mathcal{S}(\mathcal{A}(U_1,\ldots,U_r))$ is an OOS with size $r(q^{m-k}-1)/(q-1)$ and parameters $(q^m-1,q^k,q^{k-d/2})$.
    In particular, $\theta^{-1}(\mathcal{S}(\mathcal{A}(U_1,\ldots,U_r)))$ is an $(q^m-1,q^k,q^{k-d/2})$ OOC with size $r(q^{m-k}-1)/(q-1)$.
\end{theorem}
\begin{proof}
    Consider $\mathcal{A}(U_1,\ldots,U_r)$ as in Proposition \ref{prop:connOOScyclicsubspacecodes}. Note that for any $V \in \mathcal{A}(U_1,\ldots,U_r)$ we have that
    \[ |\mathcal{S}(V)|=q^k, \]
    and 
    \begin{enumerate}
        \item for every $V \in \mathcal{A}(U_1,\ldots,U_r)$  and every $\alpha \in \F_{q^m} \setminus \{0,1\}$
        \[ \lvert V \cap \alpha V \rvert \leq q^{k-d/2};\]
        \item for every $V_1,V_2 \in \mathcal{A}(U_1,\ldots,U_r)$ with $V_1\ne V_2$ and every $\alpha \in \F_{q^m}^*$
        \[ \lvert V_1 \cap \alpha V_2 \rvert \leq q^{k-d/2}.\]
    \end{enumerate}
    by Proposition \ref{prop:connOOScyclicsubspacecodes}. By Theorem \ref{thm:defOOS} we obtain the assertion.
\end{proof}

\begin{remark}
    This construction extends the idea in \cite[Section 4]{irimaugzi2022two} of using certain affine subspaces of the kernel of subspace polynomials. Indeed, in our terms, they consider optimal full-length one-orbit codes defined by the kernel of a trinomial subspace polynomial and then they prove a result of the form of Theorem \ref{th:OOCfromsubspace} for the studied cases. The problem with this approach is that, in general, it is hard to find in which extension suitable trinomials split completely; see e.g. \cite{mcguire2020some,santonastaso2022linearized}.
\end{remark}

Therefore, using large cyclic subspace codes we can construct large OOCs.
For a multi-orbit cyclic subspace code  $C \subseteq \mathcal{G}_q(m,k)$ having minimum distance $d$, the sphere-packing bound (see \cite[Theorem 2]{etzion2011error}) implies that
\[ |C|\leq \frac{\qbin{m}{k-d/2+1}{q}}{\qbin{k}{k-d/2+1}{q}}, \]
and so the number $r$ of distinct orbits in $C$ satisfies
\[ r\leq \frac{|C|}{q^m-1}\leq \frac{\qbin{m}{k-d/2+1}{q}}{(q^m-1)\qbin{k}{k-d/2+1}{q}}, \]
that is asymptotically $r$ is at most $q^{(m-k)(k-d/2+1)-m}$.
Therefore, $C$ define an $(q^m-1,q^k,q^{k-d/2})$ OOC with size at most 
\[\frac{q^{m-k}-1}{q-1} \frac{\qbin{m}{k-d/2+1}{q}}{(q^m-1)\qbin{k}{k-d/2+1}{q}}, \]
and hence asymptotically roughly $q^{(m-k)(k-d/2+1)-k}$

By using Theorem \ref{th:OOCfromsubspace}, we get that $C$  defines an $(q^m-1,q^k,q^{k-d/2})$ OOC with size $r(q^{m-k}-1)/(q-1)$. For such parameters, the Johnson bound \eqref{eq:johnsonbound} can be approximated as follows
\[J(q^m-1,q^k,q^{k-d/2}) \approx \frac{(q^{n-k})^{q^{k-\frac{d}{2}}}}{q^k},\] 
which implies that the constructions of OOCs from cyclic subspace codes cannot be either optimal or asymptotically optimal with respect to the Johnson bound.

\section{New constructions of OOCs}\label{sec:new}

In this section we will examine explicit constructions of OOCs from explicit constructions of cyclic subspace codes and we will see OOC constructions with parameters different from those in Table \ref{table:asynoptimalOOC}.
We start by describing the examples of \cite{roth2017construction,zullo2023multi}.

To this aim we need the following notation: let $\mathcal{L}_{k,q}$ denote the $\fq$-algebra of $q$-polynomials (or linearized polynomials), that is polynomials of this form
\[ \sum_{i=0}^t a_i x^{q^i} \in \F_{q^k}[x]. \]
For a linearized polynomial $f \in \mathcal{L}_{k,q}$ and an element $\xi \in \F_{q^{rk}}\setminus\F_{q^k}$ we denote
\[ W_{f,\xi}=\{x+\xi f(x) \colon x \in \F_{q^k}\}\subseteq  \F_{q^{rk}}. \]
The subset $W_{f,\xi}$ turns out to be an $\fq$-subspace of $\F_{q^{rk}}$ of dimension $k$.

\begin{theorem} [see \textnormal{\cite{roth2017construction,zullo2023multi}}] \label{thm:examplepseudo}
Let $s$ be a positive integer coprime with $k\geq 2$ and $m=2k$, let $\xi \in \fqn\setminus \F_{q^k}$ and let $f_i(x)=\mu_i x^{q^s} \in \mathcal{L}_{k}$ for $i \in \{1,\ldots,r\}$ such that $r\leq q-1$,  $\N_{q^k/q}(\mu_i)\ne\N_{q^k/q}(\mu_j)$ and $\N_{q^k/q}(\mu_i \mu_j \xi^{q^k+1})\ne 1$ for every $i\ne j$.
Then 
\[C=\bigcup_{i=1}^rC_{W_{f_i,\xi}}\] 
is a cyclic subspace code of size $r\frac{q^m-1}{q-1}$ and minimum distance $m-2$.
\end{theorem}


In the next corollary, which is from \cite[Construction 37 and Lemma 38]{roth2017construction}, the authors show a possible value for the number $r$ of orbits and a possible choice of the $\mu_i$'s.

\begin{corollary} \label{cor:subspacecodesidon}
For a prime power $q \geq 3$ and a positive integer $k\geq 2$, let $w$ be a primitive element of $\F_{q^k}$ and let $s$ be a positive integer such that $\gcd(s,k)=1$. Let $b \in \F_{q^k}$ be such that the polynomial $p(x)=x^2+bx+w$ is irreducible over $\F_{q^k}$ and $w$ is not a $(q-1)$-power in $\F_{q^k}$ (such a $b$ always exist). For $m=2k$, let $\xi\in \F_{q^m}$ be a root of $p(x)$. For $i \in \{0,1,\ldots,r-1\}$, where $r=\lfloor(q-1)/2 \rfloor$, let $\xi_i=w^i\xi$ and let 
$$V_i=W_{x^{q^s},\xi_i}=\{u+u^{q^s} \xi_i:u \in \F_{q^k}\}.$$
The set 
\[G_{m,s}=\bigcup_{i \in \{0,1,\ldots,r-1\}} C_{V_i}\subseteq \mathcal{G}_q(m,k)\] 
is a subspace code of size $r \cdot (q^m-1)/(q-1)$ and minimum distance $n-2$.
\end{corollary}

We can now use the above construction in Theorem \ref{th:OOCfromsubspace} to construct OOSs and OOCs (cf. Theorem \ref{thm:examplepseudo}).

\begin{corollary}\label{cor:constOOC2k}
Let $s$ be a positive integer coprime with $k\geq 2$ and $m=2k$, let $\xi \in \fqn\setminus \F_{q^k}$ and let $f_i(x)=\mu_i x^{q^s} \in \mathcal{L}_{k}$ for $i \in \{1,\ldots,r\}$ such that $r\leq q-1$,  $\N_{q^k/q}(\mu_i)\ne\N_{q^k/q}(\mu_j)$ and $\N_{q^k/q}(\mu_i \mu_j \xi^{q^k+1})\ne 1$ for every $i\ne j$.
Then 
\[C=\bigcup_{i=1}^rC_{W_{f_i,\xi}}.\] 
Consider now the family
\[\mathcal{A}=\mathcal{A}(W_{f_1,\xi},\ldots,W_{f_r,\xi}).\]
Then $\mathcal{S}(\mathcal{A})$ is an OOS with parameters $(q^m-1,q^k,q)$ of size $r(q^{k}-1)/(q-1)$ and\\ $\theta^{-1}(\mathcal{S}(\mathcal{A}(U_1,\ldots,U_r)))$ is an $(q^{2k}-1,q^k,q)$ OOC with size $r(q^{k}-1)/(q-1)$.\\
Choosing $\xi$ and the $\mu_i$'s as in Corollary \ref{cor:subspacecodesidon}, we have an $(q^{2k}-1,q^k,q)$ OOC with size $\lfloor(q-1)/2 \rfloor(q^{k}-1)/(q-1)$.
\end{corollary}



We can do a case-by-case analysis on the parameters of the OOCs constructed in the above proposition with those in Table \ref{table:asynoptimalOOC} and we can show that the parameters of our constructions are new.

\begin{proposition}
    The parameters of the OOCs constructed in Corollary \ref{cor:constOOC2k} are new.
\end{proposition}

We will prove the above proposition in a more general context in a few lines.

We can use the largest known constructions of cyclic subspace codes to construct large OOCs. 
In Table \ref{table:constructionscyclic} we resume the known parameters of known cyclic subspace codes. We report the family of cyclic subspace codes having the largest number $r$ of orbits, with respect to their parameters.

\begin{table}[htp]
	\centering
	\medskip
\small
\begin{tabular}{|c|c|c|c|} 
	\hline
&	Parameters & $r$ & References\\ 
	\hline
1 & $m=2k, q>2$ & $r_1=\left\lfloor \frac{q-1}{2}\right\rfloor$ & \cite[Lemma 38]{roth2017construction} \\
\hline
2& $m=ak,a>2$ & $r_2=(k-1)(q-1)\frac{q^{(e+1)k}-1}{q^k-1}, $ & \cite[Theorem 3.3]{zhang2024large} \\
& & $e=\left(\left\lceil \frac{m}{2k}\right\rceil-2\right)$& \\
\hline 
3 & $m=ak,a>2$ & $r_3=q^k(q-1)\frac{q^{(e+1)k}-1}{q^k-1},$ & \cite[Theorem 3.5]{zhang2024large} \\
& & $e=\left(\left\lceil \frac{m}{2k}\right\rceil-2\right)$& \\
\hline
4&$m \geq 6$ \mbox{ and }$k \leq \left\lfloor \frac{m-2}{4}\right\rfloor$ & 1 & \cite[Theorem 19]{roth2017construction} \\
\hline 
5 & $m \geq \binom{k+1}{2}$ & 1 & \cite[Theorem 31]{roth2017construction} \\
\hline
6 & $m =ak' $ with $a >6$ and $k=2k'$ & 1 & \cite[Theorem 3.1]{zhang2023new} \\
\hline
7 & $m \mid (q-1)$ and $k \leq \left\lfloor \frac{m+1}{4}\right\rfloor$ & 1 & \cite[Theorem 3.8]{zhang2022new} \\
\hline
8 & $m=ak(k-1)+1,$ $ k-1$ power of $p$ and $a>1$ & 1 & \cite[Theorem 7.8]{santonastaso2022linearized} \\
\hline
9 & $m=a(k^2-1)$, $ k$ and $q$ power of $2$ and $a>1$ & 1 & \cite[Theorem 7.9]{santonastaso2022linearized} \\
\hline
\end{tabular}
\caption{Table with known parameters of cyclic subspace codes}\label{table:constructionscyclic}
\end{table}

Combining Theorem \ref{th:OOCfromsubspace} and the constructions of cyclic subspace codes in Table \ref{table:constructionscyclic}, we show new examples of OOCs.

\begin{theorem}
    The cyclic subspace codes in Table \ref{table:constructionscyclic} produce OOCs with parameters as in Table \ref{table:newconstruction}.
\end{theorem}

The $i$-th row of Table \ref{table:constructionscyclic} is obtained by applying Theorem \ref{th:OOCfromsubspace} to the $i$-th row of Table \ref{table:newconstruction}.
We will now show that we obtain infinitely many new parameters.

\begin{proposition}
    The OOCs with parameters as in Table \ref{table:newconstruction} are new for any $k >2$ and $q \geq 3$.
\end{proposition}
\begin{proof}
We prove the newness of the parameters of OOCs in Table \ref{table:newconstruction} by comparing them with the ones in Table \ref{table:asynoptimalOOC}. All the OOCs in Table \ref{table:newconstruction} have parameters of the form $(q^m-1,q^k,q)$, for some $q$ prime power and $m \geq 2k \geq 6$ positive integers. 
\begin{enumerate}
    \item  Assume that $(q^m-1,q^k,q)=\left(\overline{q}^2-1,\overline{q}-2,3\right)$. By the equality of the last two parameters, we get $q=3$ and $\overline{q}=3^k+2$. Then by the equality of the first parameter, we get $3^m-1=(2+3^k)^2-1$ and by reducing modulo 3, we get a contradiction.
    
    \item  Assume that $(q^m-1,q^k,q)=\left(\frac{\overline{q}^5-1}{\overline{q}-1},\overline{q},3\right)$. By the equality of the last two parameters, we get $q=3$ and $\overline{q}=3^k$. Then by the equality of the first parameter, we get $3^m-1=\frac{3^{5k}-1}{3^k-1}$. By reducing modulo 3, we get a contradiction.
    
    \item  Assume that $(q^m-1,q^k,q)=\left(\frac{\overline{q}^{k'(\lambda+2)}-1}{\overline{q}^{k'}-1},\overline{q}-1,\lambda\right)$. By the equality of the last two parameters, we get $q=\lambda$ and $\overline{q}=\lambda^k+1$. Then by the equality of the first parameter, we get $\lambda^m-1=\frac{(\lambda^k+1)^{k'(\lambda+2)}-1}{(\lambda^k+1)^{k'}-1}$. By reducing modulo $\lambda^k$, we get $-1=\lambda+2 \ \mod \lambda^k$, a contradiction.
    
    \item  Assume that $(q^m-1,q^k,q)=\left(\frac{\overline{q}^{(h+1)}-1}{\overline{q}-1},\overline{q}+1,\lambda\right)$. Using the same argument of the previous case, we obtain again a contradiction.
    
    \item  Assume that $(q^m-1,q^k,q)=\left(\overline{q}^h-1,\overline{q}-\lambda+3,\lambda\right)$. By the equality of the last two parameters, we get $q=\lambda$ and $\overline{q}=\lambda^k+\lambda+3$. Then by the equality of the first parameters, we get $\lambda^m-1=(\lambda^k+\lambda+3)^h-1$. Since $\lambda =q>2$, by comparing the maximum degree of both side of the previous equation, we get that $h=2$. A contradiction to the fact that $h>\lambda \geq 3$.
    
    \item  Assume that $(q^m-1,q^k,q)=\left(\overline{p}h,h,\lambda\right)$. By the equality of the last two parameters, we get $q=\lambda$ and $h=\lambda^k$. Then by the equality of the first parameter, we get $\lambda^m-1=\overline{p}\lambda^k$. By reducing modulo $\lambda$, we get a contradiction.
    
    \item  Assume that $(q^m-1,q^k,q)=\left((\overline{q}-1)p,(p-\lambda),\lambda \right)$. One can argue as in the previous case.
    
    \item  Assume that $(q^m-1,q^k,q)=\left(h(\overline{q}+1),h,2\lambda\right)$. One can argue as in Case $(6)$.
    
    \item  Assume that $(q^m-1,q^k,q)=\left(\frac{\overline{q}^{\lambda-1}-1}{\overline{q}-1},\overline{q}+1,\lambda\right)$. One can argue as in Case $(3)$. 
    
    \item  Assume that $(q^m-1,q^k,q)=\left(\overline{q}^a-1,\overline{q}^h,\overline{q}^{h-1}\right)$.  By the equality of the last two parameters, we get $k=2$, a contradiction.
    
     \item  Assume that $(q^m-1,q^k,q)=\left(\overline{q}-1,f,\lambda\right)$.  By the equality of the three parameters, we get $\lambda=q$, $f=q^k$ and $\overline{q}=q^m$. Since $f \mid (\overline{q}-1)$, we have $q^k \mid (q^m-1)$, a contradiction.
     
\end{enumerate}
\end{proof}

\begin{table}[htp]
\tabcolsep=0.2 mm
\begin{tabular}{|c|c|c|c|}
\hline
 & \mbox{Parameters} & \mbox{Conditions} & Size \\ 
\hline
1 & $(q^{2k}-1,q^k-1,q)$ & $q>2$ & $r_1\frac{q^k-1}{q-1}$ \\
\hline
2 & $(q^{ak}-1,q^k-1,q)$ & $a>2$ & $r_2 \frac{q^{k(a-1)}-1}{q-1}$ \\
\hline 
3 & $(q^{ak}-1,q^k-1,q)$ & $a>2$ & $r_3 \frac{q^{k(a-1)}-1}{q-1}$ \\
\hline 
4 & $(q^m-1,q^k-1,q)$ & $m \geq 6$ \mbox{ and }$k \leq \left\lfloor \frac{m-2}{4}\right\rfloor$  & $\frac{q^{m-k}-1}{q-1}$ \\
\hline
5 & $(q^m-1,q^k-1,q)$ & $m \geq \binom{k+1}{2}$ & $\frac{q^{m-k}-1}{q-1}$ \\
\hline 
6 & $(q^m-1,q^{2k}-1,q)$  & $m =a{k/2} $ with $a >6$ & $\frac{q^{m-2k}-1}{q-1}$ \\
\hline
7 & $(q^m-1,q^k-1,q)$ & $m \mid (q-1)$ and $k \leq \left\lfloor \frac{m+1}{4}\right\rfloor$  & $\frac{q^{m-k}-1}{q-1}$ \\ 
\hline
8 & $(q^m-1,q^k-1,q)$ & $m=ak(k-1)+1,$ $ k-1$ power of $p$ and $a>1$   & $\frac{q^{m-k}-1}{q-1}$ \\ 
\hline
9 & $(q^m-1,q^k-1,q)$ & $m=a(k^2-1),$ $ k$ and $q$ power of $2$ and $a>1$  & $\frac{q^{m-k}-1}{q-1}$ \\ 
\hline 
\end{tabular}
\caption{New constructions ($q$ is a power of a prime number $p$)}
\label{table:newconstruction}
\end{table}

\section*{Acknowledgements}

The research of the last two authors was partially supported by the project COMBINE of ``VALERE: VAnviteLli pEr la RicErca" of the University of Campania ``Luigi Vanvitelli'' and was partially supported by the INdAM - GNSAGA Project \emph{Tensors over finite fields and their applications}, number E53C23001670001.

\bibliographystyle{amsplain}
\bibliography{biblio}

\providecommand{\bysame}{\leavevmode\hbox to3em{\hrulefill}\thinspace}
\providecommand{\MR}{\relax\ifhmode\unskip\space\fi MR }
\providecommand{\MRhref}[2]{%
  \href{http://www.ams.org/mathscinet-getitem?mr=#1}{#2}
}
\providecommand{\href}[2]{#2}
\begin{thebibliography}{10}

\bibitem{abel2004some}
Julian Abel and Marco Buratti, \emph{Some progress on (v, 4, 1) difference families and optical orthogonal codes}, Journal of Combinatorial Theory, Series A \textbf{106} (2004), no.~1, 59--75.

\bibitem{alderson2007optical}
Tim~L. Alderson, \emph{Optical orthogonal codes and arcs in $\mathrm{PG (d, q)}$}, Finite Fields and Their Applications \textbf{13} (2007), no.~4, 762--768.

\bibitem{alderson2007constructions}
Tim~L. Alderson and Keith~E. Mellinger, \emph{Constructions of optical orthogonal codes from finite geometry}, SIAM Journal on Discrete Mathematics \textbf{21} (2007), no.~3, 785--793.

\bibitem{alderson2007optical2}
\bysame, \emph{Optical orthogonal codes from {S}inger groups}, Advances In Coding Theory And Cryptography, World Scientific, 2007, pp.~51--69.

\bibitem{alderson2008classes}
Tim~L. Alderson and Keith~E Mellinger, \emph{Classes of optical orthogonal codes from arcs in root subspaces}, Discrete mathematics \textbf{308} (2008), no.~7, 1093--1101.

\bibitem{alderson2008families}
Tim~L. Alderson and Keith~E. Mellinger, \emph{Families of optimal {OOCs} with $\lambda= 2$}, IEEE transactions on information theory \textbf{54} (2008), no.~8, 3722--3724.

\bibitem{alderson2008geometric}
\bysame, \emph{Geometric constructions of optimal optical orthogonal codes.}, Advances in Mathematics of Communications \textbf{2} (2008), no.~4, 451--467.

\bibitem{buratti1995powerful}
Marco Buratti, \emph{A powerful method for constructing difference families and optimal optical orthogonal codes}, Designs, Codes and cryptography \textbf{5} (1995), no.~1, 13--25.

\bibitem{buratti2002cyclic}
\bysame, \emph{Cyclic designs with block size 4 and related optimal optical orthogonal codes}, Designs, Codes and cryptography \textbf{26} (2002), no.~1, 111--125.

\bibitem{buratti2010further}
Marco Buratti and Anita Pasotti, \emph{Further progress on difference families with block size 4 or 5}, Designs, Codes and Cryptography \textbf{56} (2010), 1--20.

\bibitem{charmchi2006outer}
Hossein Charmchi and Jawad~A Salehi, \emph{Outer-product matrix representation of optical orthogonal codes}, IEEE transactions on communications \textbf{54} (2006), no.~6, 983--989.

\bibitem{chu2004optimal}
Wensong Chu and Charles~J. Colbourn, \emph{Optimal (n, 4, 2)-{OOC} of small orders}, Discrete Mathematics \textbf{279} (2004), no.~1-3, 163--172.

\bibitem{chung1989optical}
Fan~RK Chung, Jawad~A Salehi, and Victor~K Wei, \emph{Optical orthogonal codes: design, analysis and applications}, IEEE Transactions on Information theory \textbf{35} (1989), no.~3, 595--604.

\bibitem{chung1990optical}
Habong Chung and P~Vijay Kumar, \emph{Optical orthogonal codes-new bounds and an optimal construction}, IEEE Transactions on Information theory \textbf{36} (1990), no.~4, 866--873.

\bibitem{ding2003several}
Cunsheng Ding and Chaoping Xing, \emph{Several classes of (2m- 1, w, 2) optical orthogonal codes}, Discrete Applied Mathematics \textbf{128} (2003), no.~1, 103--120.

\bibitem{ding2004cyclotomic}
\bysame, \emph{Cyclotomic optical orthogonal codes of composite lengths}, IEEE transactions on communications \textbf{52} (2004), no.~2, 263--268.

\bibitem{etzion2011error}
Tuvi Etzion and Alexander Vardy, \emph{Error-correcting codes in projective space}, IEEE Transactions on Information Theory \textbf{57} (2011), no.~2, 1165--1173.

\bibitem{feng2008constructions}
Tao Feng, Yanxun Chang, and Lijun Ji, \emph{Constructions for strictly cyclic 3-designs and applications to optimal oocs with $\lambda$= 2}, Journal of Combinatorial Theory, Series A \textbf{115} (2008), no.~8, 1527--1551.

\bibitem{irimaugzi2022two}
Canberk {\.I}rima{\u{g}}z{\i} and Ferruh {\"O}zbudak, \emph{On two applications of polynomials xk-cx-d over finite fields and more}, International Workshop on the Arithmetic of Finite Fields, Springer, 2022, pp.~14--32.

\bibitem{johnson1962new}
Selmer Johnson, \emph{A new upper bound for error-correcting codes}, IRE Transactions on Information Theory \textbf{8} (1962), no.~3, 203--207.

\bibitem{koetter2008coding}
Ralf Koetter and Frank~R Kschischang, \emph{Coding for errors and erasures in random network coding}, IEEE Transactions on Information theory \textbf{54} (2008), no.~8, 3579--3591.

\bibitem{mcguire2020some}
Gary McGuire and Daniela Mueller, \emph{Some results on linearized trinomials that split completely}, Finite fields and their applications. De Gruyter Proc. Math (2020), 149--164.

\bibitem{miyamoto2004optical}
Nobuko Miyamoto, Hirobumi Mizuno, and Satoshi Shinohara, \emph{Optical orthogonal codes obtained from conics on finite projective planes}, Finite Fields and Their Applications \textbf{10} (2004), no.~3, 405--411.

\bibitem{moreno1995new}
Oscar Moreno, Zhen Zhang, P~Vijay Kumar, and Victor~A Zinoviev, \emph{New constructions of optimal cyclically permutable constant weight codes}, IEEE Transactions on Information Theory \textbf{41} (1995), no.~2, 448--455.

\bibitem{omrani2005improved}
Reza Omrani, Oscar Moreno, and P~Vijay Kumar, \emph{Improved johnson bounds for optical orthogonal codes with $\lambda>1$ and some optimal constructions}, Proceedings. International Symposium on Information Theory, 2005. ISIT 2005., IEEE, 2005, pp.~259--263.

\bibitem{roth2017construction}
Ron~M Roth, Netanel Raviv, and Itzhak Tamo, \emph{Construction of {S}idon spaces with applications to coding}, IEEE Transactions on Information Theory \textbf{64} (2017), no.~6, 4412--4422.

\bibitem{ruiz2020new}
Hamilton~M Ruiz, Luis~M Delgado, and Carlos~A Trujillo, \emph{A new construction of optimal optical orthogonal codes from {S}idon sets}, IEEE Access \textbf{8} (2020), 100749--100753.

\bibitem{salehi1989code}
Jawad~A Salehi, \emph{Code division multiple-access techniques in optical fiber networks. i. fundamental principles}, IEEE transactions on communications \textbf{37} (1989), no.~8, 824--833.

\bibitem{salehi1989code2}
Jawad~A Salehi and Charles~A Brackett, \emph{Code division multiple-access techniques in optical fiber networks. ii. systems performance analysis}, IEEE Transactions on communications \textbf{37} (1989), no.~8, 834--842.

\bibitem{santonastaso2022linearized}
Paolo Santonastaso and Ferdinando Zullo, \emph{Linearized trinomials with maximum kernel}, Journal of Pure and Applied Algebra \textbf{226} (2022), no.~3, 106842.

\bibitem{weng2001optical}
Chi-Shun Weng, \emph{Optical orthogonal codes with nonideal cross correlation}, Journal of Lightwave Technology \textbf{19} (2001), no.~12, 1856.

\bibitem{yang1995some}
G-C Yang, \emph{Some new families of optical orthogonal codes for code-division multiple-access fibre-optic networks}, IEE Proceedings-Communications \textbf{142} (1995), no.~6, 363--368.

\bibitem{yang1995optical}
Guu-Chang Yang and Thomas~E Fuja, \emph{Optical orthogonal codes with unequal auto-and cross-correlation constraints}, IEEE transactions on information theory \textbf{41} (1995), no.~1, 96--106.

\bibitem{zhang2024large}
He~Zhang, Chunming Tang, and Xiwang Cao, \emph{Large optimal cyclic subspace codes}, Discrete Mathematics \textbf{347} (2024), no.~7, 114007.

\bibitem{zhang2023new}
He~Zhang, Chunming Tang, and Xing Hu, \emph{New constructions of {S}idon spaces and large cyclic constant dimension codes}, Computational and Applied Mathematics \textbf{42} (2023), no.~5, 230.

\bibitem{zhang2022new}
Tao Zhang and Gennian Ge, \emph{New constructions of {S}idon spaces}, Journal of Algebraic Combinatorics (2022), 1--14.

\bibitem{zullo2023multi}
Ferdinando Zullo, \emph{Multi-orbit cyclic subspace codes and linear sets}, Finite Fields and Their Applications \textbf{87} (2023), 102153.

\end{thebibliography}
\end{document}